\documentclass[letterpaper,twocolumn,10pt]{article}
\usepackage{usenix}

\usepackage{tikz}

\usepackage{filecontents}
\usepackage{amsmath,amssymb,amsfonts,amsthm}
\usepackage[ruled,linesnumbered]{algorithm2e}
\usepackage{graphicx}
\usepackage{textcomp}
\usepackage{xcolor}
\usepackage{hyperref}
\usepackage{xspace}
\usepackage{multirow}
\usepackage{caption}
\usepackage{subcaption}
\usepackage{soul}
\usepackage{thmtools}
\usepackage{float}

\newfloat{contract}{tbp}{loc}
\floatname{contract}{Contract}

\newfloat{circuit}{tbp}{loc}
\floatname{circuit}{Circuit}

\SetKwComment{Comment}{// }{}
\SetKw{Assert}{assert} 
\SetKwProg{Fn}{Function}{ is}{end\\}
\SetKwProg{try}{try}{:}{}
\SetKwProg{catch}{catch}{:}{end}
\SetKw{Revert}{revert}

\theoremstyle{definition}
\newtheorem{definition}{Definition}
\theoremstyle{plain}
\newtheorem{lemma}{Lemma}
\usepackage{todonotes}

\newcommand{\ru}{\textsf{Calyx}\xspace}

\soulregister{\cite}7
\soulregister{\citep}7
\soulregister{\citet}7
\soulregister{~\ref}7
\soulregister{\pageref}7
\soulregister{\vc}7

\begin{document}
\title{\Large \bf \ru: Privacy-Preserving
Multi-Token Optimistic-Rollup Protocol\\
}


\author{
{\rm Dominik Apel}\\
Common Prefix
\and
{\rm Zeta Avarikioti}\\
TU Wien \& Common Prefix
\and
{\rm Matteo Maffei}\\
TU Wien
\and{\rm Yuheng Wang}\\
TU Wien
} 

\newcommand{\mm}[1]{{\color{magenta} \textbf{Matteo:} #1}}

\maketitle

\begin{abstract}

Rollup protocols have recently received significant attention as a promising class of Layer 2 (L2) scalability solutions. By utilizing the Layer 1 (L1) blockchain solely as a bulletin board for a summary of the executed transactions and state changes, rollups enable secure off-chain execution while avoiding the complexity of other L2 mechanisms. However, to ensure data availability, current rollup protocols require the plaintext of executed transactions to be published on-chain, resulting in inherent privacy limitations. 

In this paper, we address this problem by introducing \ru, the first privacy-preserving multi-token optimistic-Rollup protocol. \ru guarantees full payment privacy for all L2 transactions, revealing no information about the sender, recipient, transferred amount, or token type. The protocol further supports atomic execution of multiple multi-token transactions and introduces a transaction fee scheme to enable broader application scenarios while ensuring the sustainable operation of the protocol. To enforce correctness, \ru adopts an efficient one-step fraud-proof mechanism. We analyze the security and privacy guarantees of the protocol and provide an implementation and evaluation. Our results show that executing a single transaction costs approximately \$0.06 (0.00002 ETH) and incurs only constant-size on-chain cost in asymptotic terms.


\end{abstract}

\section{Introduction}
\label{sec:intro}

Due to the decentralized and security-preserving properties, blockchain platforms such as Ethereum have attracted significant attention in recent years as a foundation for digital transactions and payments. As of the end of April 2025, Ether (ETH), the native token of Ethereum, had a market capitalization of approximately 211 billion \$~\cite{coinmarketcap}. Additionally, Ethereum also supports the issuance and trading of non-native tokens conforming to the ERC-20 standard~\cite{eip20}, which are widely used for decentralized applications (dApps), decentralized finance (DeFi), and decentralized autonomous organizations (DAOs) within the Ethereum ecosystem. To allow every participant to verify transaction validity and maintain a consistent view of the ledger, blockchain systems such as Ethereum publish all transactions in plaintext. This public verifiability, however, comes at a price, introducing three challenges that hinder broader adoption: privacy, scalability, and on-chain cost.

\textbf{Privacy challenge.} In blockchains such as Ethereum, plaintext
transactions expose the sender, recipient, and transferred value,
allowing on-chain activity to be linked to identities through address
clustering~\cite{victor2020address}. This linkability leaves users in
a permissionless network vulnerable to attacks such as
censorship~\cite{wahrstatter2024blockchain} and
front-running~\cite{torres2021frontrunner}. Privacy-preserving
cryptocurrencies such as Zcash~\cite{Zcash} and
Monero~\cite{alonso2020zero} mitigate linkability through
cryptographic techniques, but they operate as standalone chains with
their own consensus, and thus do not inherit Ethereum's security
guarantees. Smart-contract-based L1 solutions on
Ethereum~\cite{Zether,belling2019zeth,Tornado} avoid this drawback,
but like Zcash and Monero, they remain bound by L1 throughput and
therefore inherit the scalability challenge.

\textbf{Scalability challenge.} Because of the requirement that every
maintainer re-execute all transactions, Ethereum achieves a throughput
of only 15--30 transactions per second
(TPS)~\cite{astar2023layer2}, whereas Visa can process up to
60{,}000~TPS~\cite{VisaTXPS}. L2 protocols have emerged as a promising
direction for overcoming this limitation: they offload transaction
execution from the main chain and rely on efficient on-chain
mechanisms to verify execution correctness~\cite{gudgeon2020sok}.
Among them, Rollup protocols have received substantial attention. In a
Rollup, execution is performed solely by a designated operator (or
sequencer), and only the resulting state transition is periodically
published to the L1 blockchain as a checkpoint~\cite{arbNitro}. Unlike
sidechains, which require additional trust assumptions on their
maintainers~\cite{nick2020liquid}, Rollups thereby inherit the
security guarantees of the underlying blockchain. However, most Rollup
designs post the executed transactions on-chain in plaintext to
realize \emph{data availability}, and consequently offer no privacy
guarantees.

\textbf{On-chain cost challenge.} Reconciling Rollups with privacy
raises a third challenge: the cost of convincing the L1 of correct
execution. Among the two Rollup families, only
zk-Rollups~\cite{zksync2020} have been extended to protect transaction
privacy: protocols such as Aztec~\cite{aztec} adopt zk-SNARK
aggregation schemes~\cite{rondelet2020zecale} to conceal transaction
contents while attesting to correct execution. However, a zk-Rollup
inherently requires the on-chain smart contract to verify a validity
proof for every batch, regardless of whether any misbehavior
occurs, and on-chain proof verification is among the most expensive
operations on Ethereum. This unconditional, recurring verification
cost is fundamental to the validity-proof model. Optimistic-Rollups
such as Arbitrum~\cite{arbNitro} avoid it by verifying evidence
on-chain only when a published batch is disputed, so that the common,
honest case incurs no verification cost at all; yet no existing
optimistic-Rollup provides transaction privacy.

Given the limitations of existing solutions, a natural question
arises: \emph{Is it possible to design a privacy-preserving and
scalable payment scheme that supports multi-token transactions using
Rollup protocols, while avoiding on-chain verification cost in the
common case?}

\subsection{Related work}

\textbf{Privacy-preserving payments.} Achieving privacy while preserving the security guarantees of blockchain systems is a central research challenge and an active area of development. The mainstream approach relies on cryptographic techniques that enable the verification of transaction execution without revealing sensitive information. For example, Zerocoin~\cite{zerocoin} and Zerocash~\cite{Zerocash} use zero-knowledge proof systems to validate transactions while hiding private information, and Monero~\cite{alonso2020zero} uses ring signatures~\cite{RingConfidentialTransactions} to protect linkability. However, these systems are tailored to their own blockchains. Such a chain could serve as a sidechain bridged to Ethereum, but this alters the security assumptions, since maintaining the sidechain's security requires additional trust assumptions of its own. In Ethereum, protocols such as ZETH~\cite{belling2019zeth} use a transaction-mixing smart contract to protect linkability, while Tornado Cash~\cite{Tornado} and Zether~\cite{Zether} employ transaction encryption and zero-knowledge proofs to enable privacy-preserving transfers. These solutions enhance transaction privacy while inheriting Ethereum's native security guarantees; however, they operate entirely at the L1 level, so the inherent throughput limitations remain.

\textbf{Rollup protocols.} Based on their verification mechanism, Rollup protocols are typically categorized into two types: zk-Rollups, such as zkSync~\cite{zksync2020}, and optimistic-Rollups, such as Arbitrum~\cite{arbNitro} and Optimism~\cite{optimism2021}. As noted in prior work~\cite{rondelet2020zecale,torralba23unmasking}, zk-Rollups leverage zk-SNARK proofs mainly to realize efficient and secure state transitions~\cite{polygon2024eip4844,scroll2024blobs}, yet the executed transaction data must still be published on the L1 blockchain to ensure data availability; consequently, their privacy guarantees remain limited in practice. Recently, zk-SNARK aggregation schemes designed for privacy-preserving transactions~\cite{rondelet2020zecale} have begun to be adopted by novel zk-Rollup protocols such as Aztec~\cite{aztec}, aiming to both protect transaction privacy and reduce the on-chain verification burden. However, a zk-Rollup inherently requires the on-chain smart contract to verify a validity proof for every batch, regardless of whether the execution is correct; the on-chain verification cost is therefore incurred unconditionally and remains a fundamental limitation.

\subsection{Contribution}

In this paper, we introduce \ru, the first privacy-preserving
multi-token optimistic-Rollup protocol, which addresses the problems
above through three corresponding design choices. For the scalability
challenge, \ru follows the Rollup paradigm: transaction execution is
delegated to an L2 operator, and only succinct state updates are
published on the L1 blockchain. For the privacy challenge, \ru adopts
a UTXO-based transaction model and employs zk-SNARKs to protect
off-chain L2 transactions, revealing no information about the sender,
receiver, transferred amount, or token type. Finally, to avoid the
unconditional on-chain verification cost inherent to zk-Rollups, \ru
enforces correctness in an optimistic manner: published batches are
presumed valid, and proofs are verified on-chain only when a verifier
raises a dispute. This optimistic model is made practical by an
efficient \emph{one-step} fraud-proof mechanism, in which any
misbehavior can be refuted in a single round by revealing a constant
number of words from the published BLOB data, data that itself leaks
no private information, so that verifiability never comes at the
expense of privacy.

Conclusively, the contribution of this paper can be summarized as follows:
\begin{itemize}
    \item We propose \ru, the first privacy-preserving multi-token optimistic-Rollup protocol. \ru adopts a UTXO-based transaction model and employs zk-SNARKs to enable privacy-preserving L2 transactions. To support efficient and privacy-preserving checkpointing and verification on L1 blockchain, we introduce \emph{pointer scheme} for BLOB data and a one-step fraud-proof mechanism based on such design (Section~\ref{sec:design}).

    \item We formalize the security and privacy properties that \ru aims to achieve, and analyze its guarantees under the honest-but-curious/Byzantine model (Section~\ref{sec:analysis}).

    \item We implement a proof-of-concept prototype of \ru on the Ethereum blockchain and evaluate the associated costs, including contract deployment and transaction publishment. An asymptotic comparison with the other two Rollup protocols, Aztec and Arbitrum, is also conducted. Our evaluation shows that \ru can be deployed efficiently, with gas costs that make it a practical and feasible solution (Section~\ref{sec:implementation}).
\end{itemize}

\section{Preliminaries}

\subsection{Optimistic-Rollup}

An optimistic-Rollup protocol enables a group of clients to conduct
off-chain transactions after joining the protocol. To join, a client
deposits a specified amount of tokens into the Rollup's smart contract
deployed on the L1 blockchain; once the deposit transaction is
confirmed on-chain, a corresponding minting transaction that
initializes the client's state within the Rollup is considered valid.
Afterwards, clients submit transactions directly to the operator and
wait for execution. The operator aggregates the received valid
transactions into batches, executes them off-chain, and periodically
publishes a checkpoint on the L1 blockchain. The executed transactions
and their results are stored in a blockchain BLOB~\cite{eip4844}, whose contents are
not re-executed or verified by the L1 validators running consensus;
only a commitment to the BLOB data is recorded on-chain. To ensure
correctness, the optimistic-Rollup protocol instead relies on external
verifiers to validate the operator's published data. If a verifier
detects misbehavior by the operator, i.e., an incorrect state update,
it may submit a fraud-proof; upon successful verification of the
fraud-proof, the invalid state update is rejected. After completing
all intended off-chain transactions, a client can leave the protocol
by submitting a leave request to the operator. The operator then
executes a burning transaction reflecting the client's final state and
publishes it on the L1 blockchain. Once the burning transaction is
confirmed on-chain, the client can withdraw its remaining funds from
the Rollup smart contract, completing the exit process.

\subsection{zk-SNARK}

A zk-SNARK (\underline{z}ero-\underline{k}nowledge \underline{s}uccinct \underline{n}on-interactive \underline{ar}gument of \underline{k}nowledge)~\cite{nizkp} proof system is a \textbf{proof system} for a relation $R$, meaning that a Prover $P$ can convince a Verifier $V$ that an instance/witness pair $(x, w)$ belongs to the relation $R$. As a proof system, it satisfies the following two fundamental security properties:

\textbf{Completeness:} For any $(x, w) \in R$, the honest Prover $P$ can convince the Verifier $V$ that $(x, w) \in R$ with a proof $\pi$.

\textbf{Computational soundness:} It is computationally infeasible for any malicious Prover $P$ to convince the Verifier $V$ that some $(x', w') \notin R$ is in $R$.

In addition to these, zk-SNARKs further provide the following properties:

\textbf{Non-interactivity:} The proof is generated and sent from $P$ to $V$ without any interaction beyond this one message.

\textbf{Succinctness:} The proof $\pi$ is of constant size and can be verified in $O(1)$ time, independent of the complexity of the relation $R$.

\textbf{Zero knowledge:} The proof leaks no information about the witness $w$; the Verifier learns only that $(x, w) \in R$ for some valid $w$.

\subsection{Merkle Tree}
\label{sec:tree}
A Merkle tree~\cite{merkleTree} uses a collision-resistant hash function
$h$ to build a data structure supporting efficient membership proofs. It
is a binary tree of depth $D$ whose $2^D$ leaves store data elements; each
internal node stores the hash $h(v_l \| v_r)$ of the values $v_l, v_r$ of
its two children. The root thus acts as a commitment to all stored data,
and the membership of a data item in a leaf can be proven with an
inclusion proof of size $O(D)$. In our paper, we adopt a variant called an \emph{append-only
Merkle tree}: all leaf values are initially set to $0$, elements are
inserted into the leftmost empty leaf, and the hash function $h$ is
replaced by
\begin{equation*}
  h^*(x) =
    \begin{cases}
      0 & \text{if $x = 0\|0$}\\
      h(x) & \text{otherwise,}
    \end{cases}
\end{equation*}
which remains collision-resistant if $h$ is. Under $h^*$, every empty
subtree has value $0$ at every level, so the tree can be maintained
without storing the hashes of empty regions.

We define the \emph{front} of an append-only Merkle tree containing $n$
elements as the pair $(n, \pi)$, where $\pi$ is the inclusion proof of
the most recently inserted leaf. The front has two properties that are
central to our fraud-proof design (Section~\ref{sec:design}). First, the
front alone suffices to append further elements and compute the resulting
roots: every sibling value needed to update the path to the root is
either contained in the front or is the root of an empty subtree, which
equals $0$ by construction. Hence a party holding only this $O(D)$-size
state, and not the leaves, can correctly replay state updates; we call
such a tree a \emph{lightweight tree}. Second, given only a root $r$, one
can verify that a claimed front $(n, \pi)$ belongs to the tree with root
$r$ by checking the inclusion proof against $r$ and checking that the
proof word at the first zero bit of the binary representation of $n$
equals $0$, i.e., that the region beyond the last insertion is indeed
empty. Together, these properties allow the on-chain Judge contract to
verify a claimed intermediate tree state and recompute the updated root
from a constant-size witness, without access to the full tree.

\section{Model}
\label{sec:model}

\textbf{System Model and Assumptions.} There are three types of participants in \ru: (1) \emph{clients},
representing the users of the protocol; (2) \emph{verifiers}, who continuously monitor the blockchain to validate
the operator's published state and detect misbehavior, any client
may act as a verifier, but we separate the two roles for clarity of
presentation; and (3) the
\emph{operator}, who verifies client transactions and publishes
transaction batches to the L1 blockchain. \ru supports multiple
operators; each operator must lock a stake in the on-chain Judge
contract, which is slashed upon a successful fraud-proof and awarded
to the disputing verifier. We assume that all participants are computationally bounded and that
standard cryptographic primitives, including secure communication
channels, digital signatures, cryptographic hash functions, encryption
schemes, and zk-SNARKs, are correctly instantiated and secure.
Communication between clients and the operator is asynchronous, i.e.,
every message is eventually delivered. All participants have
uncensored access to the L1 blockchain, with a known upper bound on
the delay for reading from and writing to an idealized secure
blockchain realizing \emph{safety} and
\emph{liveness}~\cite{graf2021security}.

\textbf{Threat Model.} We analyze the security of \ru under the honest-but-curious/Byzantine
model. Participants are classified into two types. \emph{Honest}
participants follow the protocol correctly but are honest-but-curious,
meaning they may attempt to infer private information about others by
analyzing any public data leaked during protocol execution.
\emph{Byzantine} participants may behave arbitrarily and deviate from
the protocol in an adversarial manner, while also attempting to
reconstruct private information held by honest parties. We assume the
presence of at least one honest
verifier who monitors protocol execution and responds in a timely
manner.

\textbf{Protocol Goals.} Under the Byzantine security model, we first consider \textit{balance security} in order to capture the fundamental requirement that any honest participant must not lose any tokens they are entitled to after leaving the protocol, regardless of the behavior of other parties. We define this property formally as follows:

\begin{definition}[Balance security]
\label{def:security}
Any participant honestly following \ru does not lose coins.
\end{definition}

A protocol that does not execute any transactions can trivially satisfy \textit{balance security}; however, such a protocol is functionally meaningless. To ensure that the protocol makes progress in response to valid requests from honest participants, we define the notion of \textit{liveness} as follows:

\begin{definition}[Liveness]
\label{def:liveness}
Any valid request from participants honestly following \ru will eventually either be committed on-chain or invalidated.
\end{definition}

In addition to \textit{balance security} and \textit{liveness}, which ensure that the protocol functions correctly in response to actions taken by honest participants, we define the privacy property that \ru aims to achieve as \textit{L2 transaction privacy}, meaning that only indistinguishable public leakage about L2 execution is observable through the published data on L1 blockchain.
Our \emph{L2 transaction privacy} property is defined based on the Layer 2 Indistinguishability (L2-IND) game inspired by~\cite{Zcash}. The formal definition of the game and how the oracle works can be found in Appendix~\ref{apx:privacy}.

\begin{definition}[L2 transaction privacy]
A rollup protocol realizes \emph{L2 transaction privacy} if for any polynomial-time (PPT) adversary $\mathcal{A}$, the probability of winning the L2-IND game is only negligibly greater than 1/2.
\end{definition}

\begin{figure*}[t]
    \centering
    \includegraphics[width=\textwidth]{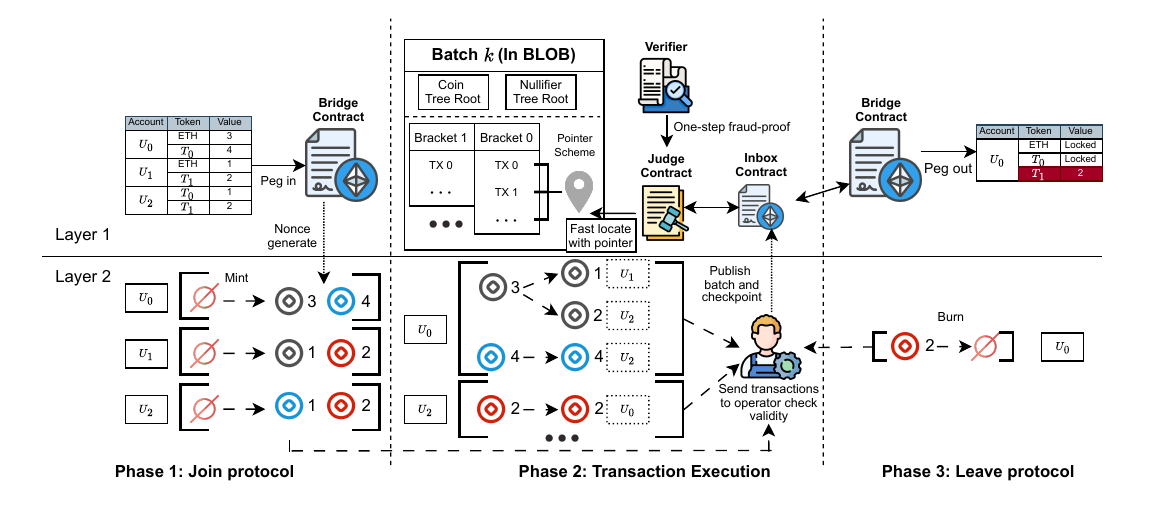}
    \caption{\ru Overview. Consider three clients $U_0, U_1, U_2$ joining the protocol with three types of tokens (ETH, $T_0$, $T_1$). Then $U_0$ privately pays $U_1$ and $U_2$, while $U_2$ pays $U_0$. Transactions grouped in a bracket are executed atomically. Execution occurs after the transactions are published on-chain and no fraud-proof is submitted by a verifier within a specified time window. Finally, $U_0$ retrieves the $2$ tokens $T_1$ from $U_2$ through the contract.}
    \label{fig:overview}
\end{figure*}

\section{Protocol Overview}

Before presenting the detailed design of \ru, we first provide a high-level overview of the protocol. Figure~\ref{fig:overview} illustrates the structure of \ru. We then describe the iterative development process, starting from a baseline construction of a standard optimistic-Rollup protocol.

\textbf{Naive design.} In a standard optimistic-Rollup, after joining the protocol clients then propose plaintext transactions in the form of regular account-based transaction ``$A$ sends $x$ coins to $B$". These transactions are submitted to an operator, published on the L1 blockchain, verified by verifiers, and finally committed. Clients can later exit the protocol with their latest state.

While this approach achieves \emph{balance security} and
\emph{liveness}, it fails to provide any privacy guarantees:
transaction details are leaked through the batches published by the
operator, whose disclosure is precisely what enables verifiers to
check correctness and uphold balance security.

\textbf{L2 with UTXO model.} Since publishing transactions on the L1 blockchain is essential for ensuring the security and verifiability of the optimistic-Rollup protocol, the main challenge is to minimize the information leakage from these published transactions. The UTXO model discloses only the linkability between inputs and outputs. Therefore, \ru adopts the UTXO model for transactions within the L2 system. Furthermore, each transaction includes a token type parameter and a transaction fee, enabling \ru to support multi-token transactions while ensuring the sustainable operation of the protocol. Finally, a corresponding mint and burn UTXO transaction is designed to interact with the L1 blockchain for client joining and leaving.

However, the validity of UTXO transaction is usually achieved with digital signatures. Nonetheless, such signatures still reveal linkability through the associated public keys used for verification, thereby leaking metadata that can be exploited by adversaries.

\textbf{Zero-Knowledge verification.} To protect linkability while ensuring that only the coin owner can spend it, we first hide the source of the input coin using a commitment, and then adopt a zk-SNARK proof system to prove transaction validity. The key distinction in applying this idea to an L2 protocol is that, unlike L1 protocols that achieve consensus directly among all participants, the optimistic-Rollup relies on the L1 blockchain as the consensus layer. Consequently, the public inputs required for verifying zk-SNARK proofs must be published on the L1 blockchain by the operator along with the periodic checkpoints. Importantly, publishing these inputs and outputs on L1 does not undermine privacy. \ru is UTXO-based and maintains no account-balance state, so no user balances are ever written to L1; the smart contract persistently stores only the Merkle roots of the Coin and Nullifier trees as the committed L2 state. These roots are not auxiliary secrets, but a deterministic digest of the privacy-preserving transactions in the batch BLOB, from whose inputs and outputs they can be fully reconstructed. The transaction data itself consists solely of coin commitments, serial numbers (nullifiers), and zk-SNARK proofs, each a hiding commitment or a pseudorandom value, and therefore reveals no sender, recipient, transferred amount, or token type. Since the roots are merely a digest of this already-hiding data, publishing them leaks nothing further. We formally capture the indistinguishability of this public leakage through our privacy definition.

Considering the high transaction throughput that a Rollup protocol is expected to handle, a new challenge arises: how to efficiently verify the data published by the operator on the L1 blockchain. This includes not only ensuring correctness and availability but also enabling timely verification to support latency-sensitive applications.

\textbf{Efficient One-step fraud-proof.} To address the computational overhead and verification delays introduced by directly combining the optimistic-Rollup framework with a ZK-proof system, this paper proposes an efficient one-step fraud-proof scheme.

Concretely, instead of relying on multiple rounds of interactive verification to locate a faulty execution within the entire published execution history in the BLOB, as is common in standard optimistic-Rollup protocols, our scheme employs a \emph{pointer} mechanism. The fraud-proof transaction includes a pointer that directly identifies the erroneous location in the operator's published BLOB data. This design eliminates the need for the smart contract to conduct extensive on-chain searches over large datasets, thereby reducing computational overhead, gas costs, and overall latency.

\section{Protocol Design}
\label{sec:design}

In this section, we present the detailed design of the \ru protocol. The protocol behavior is described from the perspective of a client's life cycle, consisting of three phases: \emph{protocol joining}, \emph{transaction execution}, and \emph{protocol leaving}. We assume that the on-chain smart contract for the protocol has been deployed in advance and is accessible to all participants. The precise structure of the contract and the interaction semantics will be formally defined in conjunction with the explanation of each phase of the protocol.


For a better understanding of our protocol design, we first list the notations used in the protocol construction, as shown in the following table~\ref {tab:notation}.

\begin{table}[!h]
\caption{Notation explanation}
\label{tab:notation}
\resizebox{\columnwidth}{!}{%
\begin{tabular}{|c|c|}
\hline
\textbf{Notation}      & \textbf{Explanation}                        \\ \hline
$h$                    & hash function                               \\ \hline
$Enc$                  & Encryption function                         \\ \hline
\begin{tabular}[c]{@{}c@{}}$sk_{coin}$\\ $pk_{coin}=h(sk_{coin})$\end{tabular} &
  \begin{tabular}[c]{@{}c@{}}secret key and public key to spend a coin\\ generated with hash function\end{tabular} \\ \hline
$sk_{Enc}$, $pk_{Enc}$ & secret key and public key for encryption    \\ \hline
$p$                    & random identifier                           \\ \hline
$token$                & token type                                  \\ \hline
$value$                & transaction value                           \\ \hline
$sn$                   & coin serial number                          \\ \hline
$crt$                  & coin Merkle tree root                       \\ \hline
$id_{L1}$              & L1 address                                  \\ \hline
$pk_{sig}$             & public key for signing                      \\ \hline
$fee$                  & L2 transaction fee provided by participants \\ \hline
\end{tabular}%
}
\end{table}

\subsection{Joining Protocol} 

As a prerequisite for the subsequent procedures, a client can join \ru in two steps: (1) deposit coins on the L1 blockchain; (2) the initial L2 state corresponding to the client is created.

For the first step, a client sends a denomination $value$ of some ERC-20 compliant token type $token$ (including the native token ETH) as initial asset and the corresponding public key used for creating signatures in the following procedures $pk_{\mathsf{sig}}$ to the on-chain smart contract.
The contract will lock these deposit assets, record the corresponding information on-chain, and reply with a fresh nonce value $n$, which can be used as evidence for creating the L2 initial state in the second step. The on-chain smart contract for the joining process is detailed in Algorithm~\ref{alg:bridge_l1_l2}. Note that in the following, multiple contract codes are presented. In practice, they can either be deployed separately or integrated with each other.

For the second step, the initial state of \ru for a client is created by generating and executing a special UTXO transaction called a minting transaction, mint. Since it initializes the off-chain state for the client, the minting transaction has no inputs. Its output consists of a single value
$c = h(token || value || fee || k)$, where $k=h(p||pk_{coin})$ is the compressed client identity and $p$ is the randomly generated coin identifier, and $pk_{coin}$ represents the owner of the coin who knows the preimage $sk_{coin}$.
The body of the transaction includes two critical pieces of information for identifying and verifying its validity: the mint nonce $n$ obtained from the previous on-chain deposit step, and the compressed identity $k$. To ensure \emph{economic sustainability} and guarantee the rational security of the protocol, we introduce a transaction fee scheme: the $fee$ in the coin commitment denotes the L2 transaction fee paid to the operator as compensation for the gas costs of publishing transactions on the L1 blockchain. These fees are included in the collateral deposit beforehand. 

Once a valid minting transaction is constructed by the client and executed according to the procedure described later, the client is considered to have successfully joined \ru. Note that the validity of a mint transaction is not verified using zk-SNARK proof; instead, it is checked directly against the records maintained in the smart contract. 

\begin{algorithm}[h]
\caption{L1 to L2 Bridge (Contract 1)}\label{alg:bridge_l1_l2}

Initialize nonce value $curNonce = 0$\;
Record nonce mapping $nonces = \mathbb{N} \to (\mathbb{N}, \mathcal{T}, \mathcal{V}, \mathcal{G}, \mathcal{P})$
\\ \Comment{(public key, token type, token value, L2 transaction fee, included block number)}
\Fn{$toL2(v, t, pk_{sig}$)}{
    $g = msg.value$\;
    \If{$t.allowance(id_{L1}, this) < v$}{
        \Revert{}
    }
    $t.transferFrom(id_{L1}, this, v)$\;
    $curNonce \gets curNonce + 1$\;
    $nonces[curNonce] \gets (pk_{sig}, t, v, g, block.number)$\;
    \Return{$n:=curNonce$}\;
}

\Fn{$feeToL2(pk_{sig})$}{
    $g = msg.value$\;
    $curNonce \gets curNonce + 1$\;
    $nonces[curNonce] \gets (pk_{sig}, 0, 0, g, block.number)$\;
    \Return{$curNonce$}\;
}

\Fn{$getMintData(n)$}{
    \If{$nonces[n] = \bot$}{
        \Return{$\bot$}\;
    }
    \Return{$nonces[n]$}\;
}
\end{algorithm}

\subsection{L2 Transaction Execution}

As an L2 protocol, \ru's primary role is to offload the execution of transactions to the off-chain environment. Here, we describe the transaction execution procedure from the perspective of a regular \emph{transfer transaction}, which also applies to special transactions such as the mint transaction described above.

As a UTXO-model transaction, a transfer transaction in \ru contains two lists of $M$ inputs and $M$ outputs (with unused elements left empty), all of the same token type. The output of the transaction is the coin commitment $c = h(token||value||fee||k=h(p||pk_{coin}))$ of the following information: (1) \textbf{token type} $token$; (2) \textbf{coin value} $value$ for the output coin; 
(3) \textbf{fee value} $fee$, measured in the native currency of the L1 chain;
(4) \textbf{identifier} $p$, which is generated randomly; (5) \textbf{coin owner public key} $pk_{coin}$. If the sender and receiver are not able to reach agreement themselves about the private value, all the information can be sent to the receiver through an encrypted message privately with the receiver's encryption key pair $C_{inf}=Enc_{pk^{receiver}_{Enc}}(token,value,fee,p)$ and attached to the coin commitment. 

And accordingly each input includes the privacy-preserving information of the coin it uses: 
\begin{itemize}
    \item The \textbf{coin tree root} $crt$ is the Merkle tree root that represents all coins ever created during \ru's lifecycle.
    \item The unique \textbf{serial number} $sn=h(p||sk_{coin})$ is unique for every coin that is only known to the true owner of the coin.
    \item The \textbf{signature public key} $pk_{sig}$ that belongs to the owner of the input coin who knows the corresponding secret key.
    \item The \textbf{coin commitment} $cm = h(token || value\allowbreak || fee\allowbreak || pk_{auth})$ to the coin's private properties, which includes: (1) \textbf{Signing Key Authorizer} $pk_{auth}=h(sk_{coin}||pk_{sig})$; (2) \textbf{Token Type} $token$; (3) \textbf{Coin Value} $value$ used for transaction; 
    (4) \textbf{Fee Value} $fee$.
\end{itemize}

To further support clients in transferring multiple types of tokens in different UTXO transactions at the same time. We further extend the idea of transaction into transaction brackets. Transactions are put into transaction brackets, acting as an "all-or-nothing" primitive, allowing the atomic execution of multiple transactions.
One of the use cases for such an atomic primitive is to swap different types of tokens. Furthermore, a hash of all the transactions' hash values is contained in the bracket to guarantee integrity. Finally, a transaction bracket must contain a signature from every public key contained in an input on the bracket hash, i.e., all coin owners must agree that their coins are spent in this specific way.



By leveraging a commitment scheme, the source of the spent coin is hidden, thereby protecting against linkability. However, this introduces a new challenge: how to ensure transaction correctness while keeping data hidden. Specifically, two correctness properties must be satisfied:
\emph{Input correctness} and \emph{Input-output relation correctness}.

\subsubsection{zk-SNARK Circuits}

In \ru, we leverage the zk-SNARK proof, which is included in the transfer transaction body, to enforce both properties mentioned above. In the following, we present the corresponding circuits that realize these two correctness goals.

\paragraph{Input proof.}Here we further define the \emph{input correctness} should realize the following goals: 
\begin{enumerate}
    \item There exist such outputs that are selected to spend from. This is done by privately providing a Merkle tree proof that asserts that the coin is a leaf in the tree with the provided root. The Merkle tree depth must be a fixed constant, called $D$.

    \item  The prover is the owner of the coin and knows the $sk_{coin}$.

    \item The provided serial number was calculated correctly.
    

    \item \emph{The owner approves of the provided signature verification key:}
    This is done by deriving a $pk_{auth}$ value from both the signature verification key and the secret key of the coin.
    This calculated value is private in order to achieve the hiding property of the provided coin commitment.

    \item \emph{The provided commitment was constructed correctly.}
    
\end{enumerate}

The inputs for the input proof's zk-SNARK circuit are as follows, with public values \underline{underlined}, all values of scalar type \textit{field} except otherwise noted, and $T[N]$ denoting an array type containing $N$ elements of type $T$:
\begin{flalign*}
&(\underline{crt},\; \underline{sn},\; \underline{cm},\; \underline{pk_{sig}},\; pk_{auth},\; c, \; bool[D]~dir,\; field[D]~path,\; r,\\
&\; token,\; value,\; fee,\; p,\; pk_{coin},\; sk_{coin}) &
\end{flalign*}

The circuit of the zk-SNARK can be described by Algorithm \ref{alg:zk-SNARK_coin_input}:

\begin{algorithm}
\caption{Input Proof (zk-SNARK Circuit 1)}\label{alg:zk-SNARK_coin_input}
\KwData{$\underline{crt} \in \mathbb{F}, \underline{sn} \in \mathbb{F}, \underline{cm} \in \mathbb{F}, \underline{pk_{sig}} \in \mathbb{F}, pk_{auth} \in \mathbb{F}, c \in \mathbb{F}, dir \in \mathbb{B}^D, path \in \mathbb{B}^D, token \in \mathbb{F}, value \in \mathbb{F}, fee \in \mathbb{F}, p \in \mathbb{F}, pk_{coin} \in \mathbb{F}, sk_{coin} \in \mathbb{F}$}\Comment{Finite field $\mathbb{F}$, Binary Set $\mathbb{B}$}
\Assert{$c = h(token || value || fee || h(p || pk_{coin}))$}\;
\Assert{$\underline{sn} = h(p || sk_{coin})$}\;
\Assert{$pk_{coin} = h(sk_{coin})$}\;
\Assert{$pk_{auth} = h(sk_{coin}||\underline{pk_{sig}})$}\; 
\Assert{$\underline{cm} = h(token || value || fee || pk_{auth})$}\;
$tmp \gets c$\;
\For{$i=1$ \emph{\KwTo} $D-1$}{
    \eIf{$dir[i]$}{
        $tmp \gets h(tmp || path[i])$\;
    }{
        $tmp \gets h(path[i] || tmp)$\;
    }
}
\Assert{$tmp = \underline{crt}$}\;
\end{algorithm}

\paragraph{Transaction proof.}We use a transaction proof to define the correct relationship between the transaction outputs and inputs. The transaction proof aims to guarantee the following goals: 

\begin{enumerate}
    \item The prover knows the secret preimage of the commitment values.

    \item All the inputs and outputs of a transfer transaction should belong to the same token type. The number of inputs and outputs is fixed to a constant number $M$.

    \item The sum of the input values should match the sum of the output values.

    \item In order to use the zk-SNARK scheme, we further require that the input and output are positive.
\end{enumerate}



The inputs are (public values are \underline{bold}):
\begin{flalign*}
&(field[M]~\underline{C},\; bool[M]~\underline{isConnected},\; bool[M]~\underline{isInput},\; field~\underline{fee}, &\\
&\; field[M]~token,\; field[M]~value,\; bool[M]~valueDecomp, &\\
&\; field[M]~fee, \; bool[M]~valueDecomp,\; field[M]~cmr) &
\end{flalign*}

Here, the two boolean vectors $isConnected$ and $isInput$ of size $M$ have the following semantics:
If $isConnected[k] = \text{False}$, then both the $k^{th}$ input and output are empty. Otherwise, $isInput[k]$ determines whether slot $k$ is connected to an input or an output. This affects the sign of the value and fee of the coin when computing the transaction balance: inputs add value to the balance, while outputs subtract value. The provided inputs and outputs for a transaction must follow the structure specified by $isConnected$ and $isInput$ for the transaction to be valid. Finally, the circuit checks that the balances sum to zero, ensuring that the total output values match the total input values. The zk-SNARK circuit can be described by Algorithm~\ref{alg:zk-SNARK_tx}.


\begin{algorithm}[!h]
\caption{Transaction Proof (zk-SNARK Circuit 2)}\label{alg:zk-SNARK_tx}
\KwData{$\underline{C} \in \mathbb{F}^M, \underline{isConnected} \in \mathbb{B}^M, \underline{isInput} \in \mathbb{B}^M, \underline{fee} \in \mathbb{F}, token \in \mathbb{F}^M, value \in \mathbb{F}^M, valueDecomp \in (\mathbb{B}^l)^M, fee \in \mathbb{F}^M, feeDecomp \in (\mathbb{B}^l)^M, cmr \in \mathbb{F}^M$} \Comment{Finite field $\mathbb{F}$, Binary Set $\mathbb{B}$}
$bal_{v} \gets 0$\;
$bal_{g} \gets 0$\;

\For{$i=0$ \KwTo $M-1$}{
    \If{$\underline{used}[i]$}{
        \Assert{$\underline{C}[i] = h(token[i] \parallel value[i] \parallel fee[i] \parallel cmr[i])$}\;
        \eIf{$\underline{isInput}[i]$}{ \tcp{This is an input}
            $bal_{v} \gets bal_{v} + value[i]$\;
            $bal_{g} \gets bal_{g} + fee[i]$\;
        }{ \tcp{This is an output}
            $bal_{v} \gets bal_{v} - value[i]$\;
            $bal_{g} \gets bal_{g} - fee[i]$\;

            \tcp{Check range}
            $value' \gets 0$\;
            $fee' \gets 0$\;
            \For{$j=0$ \KwTo $l-1$}{
                \If{$valueDecomp[i][j]$}{
                    $value' \gets value' + 2^j$\;
                }
                \If{$feeDecomp[i][j]$}{
                    $fee' \gets fee' + 2^j$\;
                }
            }
            \Assert{$value' = value[i]$}\;
            \Assert{$fee' = fee[i]$}\;
        }
    }

    \eIf{$token[i] = 0$}{
        \tcp{A fee token must not hold value}
        \Assert{$value[i] = 0$}\;
    }{
        \tcp{There can only be one non-universal token type}
        \Assert{$token[i] = token[0]$}\;
    }
}

\Assert{$bal_{g} \gets bal_{g} - \underline{\text{fee}}$}\;
\Assert{$bal_{g} = 0$}\;
\Assert{$bal_{v} = 0$}\;

\end{algorithm}

\subsubsection{Transaction execution with one-step Fraud-proof}

After receiving transactions from the client, the operator first verifies the proofs generated by the zk-SNARK scheme based on the circuit defined above. In order to execute these valid transactions, two steps to follow in \ru: (1) the operator needs to publish the valid transactions it wants to execute to the L1 blockchain, including the BLOB data and the smart contract stored data; (2) the published transactions are only considered to be executed after there exists no fraud-proof for a certain time period.
\vspace{-1em}
\paragraph{On-chain published data.}Once the operator has completed a batch of transaction brackets pending execution, it first generates a \emph{fee-collecting transaction} to collect all the fees included in the transactions and appends it to the end of the batch, This transaction contains an output with commitment $h(0||0||ck_f||k)$, where $ck_f$ is the transaction fee checkpoint at the beginning of the batch, and $k$ is a random value chosen by the operator and included in the body of the transaction for verification. Then a L1 transaction carrying this batch as BLOB data is then submitted to the Inbox contract on the L1 blockchain.

In addition to the data written to the BLOB, the transaction records metadata in the Inbox contract, including the commitment to the BLOB data (serving as the checkpoint of the L2 state) and information about the coins burned in this batch for clients leaving the protocol. Specifically, an array of elements of the form $(txbid, txid, t, v, g, id_{L1})$ is passed to the Inbox contract. Each element represents a claim that the transaction at index $txid$ within the transaction bracket at index $txbid$ is a burn transaction that burns $v$ tokens of type $t$ and $g$ units of ETH, while specifying $id$ as the receiver address. Algorithm \ref{alg:inbox} shows the working of the Inbox contract.

\begin{algorithm}
\caption{Inbox Contract (Contract 2)}\label{alg:inbox}
\KwData{Fraud-proof Period $fpp: \mathbb{N}$}
\Comment{Persistent storage}
$batchCommitments = \mathbb{N} \to (\mathcal{C}, \mathcal{I})$\;
$curHeight = \mathbb{N}$\;
$hnb = \mathbb{N}$; \Comment{highest non-final batch}
$blockFin = \mathbb{N} \to \mathbb{N}$\;
$burnData = \mathbb{N} \to (\mathbb{N}^2 \to \mathcal{B})$\;

\Fn{newBatch($blockBurnData:(\mathbb{N}^2\to\mathcal{B}), prevBatch: \mathcal{B}$)}{
    \Assert{$Judge.hasStaked(msg.sender)$}\;
    \Assert{$batchCommitments[curHeight][0] == prevBatch$}\;
    \While{$(blockFin[hnb] < block.number) \wedge (hnb \leq curHeight)$}{
        $blockFin[hnb] \gets \bot$\;
        $hnb \gets hnb +1$\;
    }

    $curHeight \gets curHeight + 1$\;
    $bh \gets blobhash(0)$\;
    \Comment{Force exactly one BLOB attached}
    \Assert{$bh \neq 0 \wedge blobhash(1) = 0$}\;
    $batchCommitments[curHeight] = (bh, msg.sender)$ \;
    $burnData[blockHeight] = blockBurnData$\;
    $blockFin[curHeight] \gets block.number + fpp$\;
}

\Fn{consumeBurnData($batchId, bracketId, txId$)}{
    \Assert{$msg.sender = L2\_to\_L1\_Bridge$}\;
    \Assert{$burnData[batchId][bracketId][txId] \neq \bot$}\;
    $tmp \gets burnData[batchId][bracketId][txId]$\;
    $burnData[batchId][bracketId][txId] \gets \bot$\;
    \Return{$tmp$}\;
}

\Fn{blockFinalized($batchId$)}{
    \Return{$hnb > batchId$}\;
}

\Fn{revertL2Chain($height:\mathbb{N}$)}{
    \Assert{$msg.sender = Judge$};
    \Assert{$blockFin[height] \neq \bot$}\;
    \Assert{$blockFin[height] \geq block.number$}\;
    \For{$i=height$ \emph{\KwTo} $curHeight$}{
        $blockFin[i] \gets \bot$\;
        $blockCommitments[i] \gets \bot$\;
        $burnData[i] \gets \bot$\;
    }
    $curHeight = height - 1$\;
}   

\end{algorithm}


Two types of data are recorded in the blockchain BLOB part. First, at the beginning of the batch, the updated root of the \textbf{Coin Tree}, which records all coins that have ever existed, and the root of the \textbf{Nullifier Tree}, which records all coins that have been spent, are stored as the current L2 state. The Merkle roots of these two trees can then be used to efficiently prove or dispute the validity of a coin. Note that the leaves of the tree, i.e., the coin commitments, are also published in the BLOB. This allows the construction of the tree path and the generation of the corresponding Merkle proof. Additionally, a transaction fee checkpoint $ck_{f}$ is also at the beginning of the batch to record transaction fee collection.

The second part is the executed transaction brackets batch. Given the high transaction throughput and the fraud-proof scheme used to guarantee the security of \ru, the goal here is to carefully define the published data structure so that a fraud-proof can be verified by referencing only a small, constant number of data words in BLOB. In \ru, each published batch begins with a header containing the number $m$ of transaction brackets, followed by $m$ pointers to the beginning of the respective transaction
bracket $tx_i$. A bracket serialization then begins with the bracket hash and the root of the coin tree after its execution, followed by the number of transactions and the number of required signatures.

To record transactions in a way that enables efficient access to any part of the transaction information, we first serialize the Inputs and Outputs, defined as the concatenation of the public information for inputs $s_{input}(I)$ and outputs $s_{output}(O)$ introduced earlier. The input information includes: the Coin Tree root $crt$, the serial number of the coin $sn$, the coin commitment $cm$, the public key for signatures $pk_{sig}$, and the zk-SNARK proof $\pi$. The output information includes: a bit indicating whether the serialization contains encrypted extra data about the private value, the coin commitment of this output $c$, and optional encrypted extra data $C_{inf}$. With such serialization, each transaction within a bracket begins with its transaction hash, which is defined as $h(s_{input}(I_0)||s_{input}(I_1)||\cdots || s_{output}(O_0)||s_{output}(O_1)||\cdots||fee||\allowbreak proof)$, ensuring that the commitment captures both the public components and the zero-knowledge proofs.

The transaction published along with the batch contains all public data required for verification. Each component is assigned a fixed byte length to enable efficient access under the pointer scheme. The transaction first specifies its token type, the number of inputs and outputs, pointers to the serialized inputs and outputs, and a reference that authenticates the coin tree root used by the input proof. This is followed by the fee field, the transaction proof, the serialized inputs and outputs, and finally the transaction signature.

For transfer transactions, the serialization follows this complete structure, explicitly specifying the number of inputs and outputs and including the transaction proof. In contrast, mint, burn, and fee-collecting transactions use implicit input or output structures. These transactions do not carry a transaction proof in the same form; instead, they include revealed inputs or outputs, such as a revealed output in a mint transaction or a revealed input in a burn or fee-collecting transaction.

\vspace{-1em}

\paragraph{On-chain fraud-proof.}The verifier must generate a fraud-proof based on the data published by the operator. Leveraging the published data structure described above, the verifier can pinpoint a specific error by revealing only a small, constant number of BLOB data words to the Judge contract. The Judge contract first checks that these revealed data words are indeed part of the BLOB at the claimed positions by verifying them against the BLOB commitment recorded in the Inbox contract. It then parses the fraud-proof using the provided data, which allows the Judge contract to validate the proof. If the proof is valid, the verifier may claim the operator’s stake as both punishment for the dishonest publisher and a reward for successful verification. The transactions published in the BLOB are considered executed only if no fraud-proof challenge is raised by a verifier within a specified time window.

We enumerate all possible violations of the protocol that could impact the benefits of honest clients as a set of rules from two perspectives: (1) \emph{semantic validity}, i.e., the published data structure correctly follows the specified format; and (2) \emph{proof validity}, i.e., validity proofs such as signatures, hash values, and zk-SNARK proofs can be correctly verified. The verifier can trigger the appropriate rule by submitting the corresponding fraud-proof evidence, requiring only a single round of communication with the Judge contract, rather than multiple rounds to localize the violation. We classify all rules into three categories:

\begin{algorithm}
\caption{Judge Contract Pseudocode (Contract 3)}\label{alg:judge}
\KwData{Required minimal stake $S$, Fraud-proof Period $fpp$}

$stake = (\mathcal{ID} \to \mathbb{N})$\;
$unstakeRequest = \mathcal{ID} \to (\mathbb{N} \to \mathbb{N})$\;

\Fn{hasStaked(id)}{
    \Return{$stake[id] \geq S$}\;
}

\Fn{stake()}{
    $stake[msg.sender] \gets stake[msg.sender] + msg.value$\;
}

\Fn{unstakeRequest()}{
    $tmp = stake[msg.sender]$\;
    \Assert{$tmp > 0$}\;
    $unstakeRequest[msg.sender][block.number + fpp] \gets tmp$\;
    $stake[msg.sender] \gets 0$\;\Comment{Can't retrieve stake before next fraud-proof period}
}

\Fn{unstake(sid)}{
    $tmp = unstakeRequest[(msg.sender, sid)]$\;
    \Assert{$tmp > 0$}\;
    \Assert{$sid \geq block.number$}\;
    $unstakeRequest[msg.sender][sid] \gets 0$\;
    $msg.sender.transfer(tmp)$\;
}

\Fn{slash(loser, winner)}{
    $loserStake \gets stake[loser]$\;
    \ForEach{$key \in unstakeRequest[loser]$}{
        $loserStake \gets loserStake + unstakeRequest[loser][key]$\;
        $unstakeRequest[loser][key] \gets 0$\;
    }
    $stake[loser] \gets 0$\;
    $winnerReward \gets loserStake$\;
    $winner.transfer(winnerReward)$\;
}

\Fn{disputeBlock(bID, data, indizes, disputeType, auxData)}{
    $checkFraudProof(bID, data, indices, disputeType,\allowbreak auxData)$;\Comment{According to the rules} 

    $Inbox.revertL2Chain(bID)$\;
    $slash(id_S, msg.sender)$\;
}
\end{algorithm}

\begin{enumerate}
    \item Rules about the validity of \textbf{transactions}
    \begin{enumerate}
        \item \label{rule:tx_type} Invalid transaction type.
        A transaction's type is not one of the three allowed types: mint, transfer, and burn. 
        \item \label{rule:tx_mint_values} Invalid mint transaction commitment. Can be disputed by recomputing commitment.
        \item \label{rule:tx_burn_values} Invalid burn transaction commitment. Can be disputed by recomputing commitment.
        \item \label{rule:tx_input_zk-SNARK} Invalid input zk-SNARK proof. Can be disputed by verifying a specific input $I_k$'s zk-SNARK proof $\pi_{I;k}$ with published parameters. 
        \item \label{rule:tx_proof_zk-SNARK}Invalid transaction zk-SNARK proof invalid. Can be disputed by verifying transaction proof $\pi_{tx}$ with the published parameters. 
        \item \label{rule:tx_hash} Invalid transaction hash. Can be disputed by recomputing the hash value. 
        
        \item \label{rule:tx_double_past} Transaction input or mint nonce double spending (not in the same bracket). Can be disputed by providing the double-spending serial number or bridge nonce $s$, the previous nullifier tree root $ntr$, and providing a Merkel tree inclusion proof that $s \in ntr$.
        
        \item \label{rule:tx_double_same} Transaction input or mint nonce double spends (in the same bracket). 
            Can be disputed by providing two serial numbers $s, s'$ in two inputs of transactions in the same bracket, and checking $s = s'$.
        \item \label{rule:tx_input_ctr}
        Invalid coin tree root. 
            Can be disputed by providing the coin tree root $ctr'$ referenced by the revealed transaction input $I_k = (\ldots, ctr, \ldots)$, and checking that $ctr \neq ctr'$. 
    \end{enumerate}
    \item Rules about the validity of transaction \textbf{brackets}
    \begin{enumerate}
        \item \label{rule:txb_hash} Invalid TX Bracket hash.
            Can be disputed by computing with all transaction hashes.
            
        \item \label{rule:txb_sig_count} Mismatch between input number and signature number. 
        Can be disputed by providing the number of inputs and the signature number for checking.
        
        \item \label{rule:txb_sig_error} Invalid signature error.
        Can be disputed by verifying the signature $\sigma_k$ with the public key $pk_k$ of the input on the bracket hash.
        
        \item \label{rule:txb_tx_count} Empty or oversize transaction bracket. 
            Can be disputed by providing the transaction count field for this transaction bracket.
    \end{enumerate}
    \item Rules about the validity of transaction \textbf{batches}
    \begin{enumerate}
        \item \label{rule:b_burn_copy} Burn transaction data on the contract does not correspond to the burn transaction in the batch.
        Can be disputed by providing the burn transaction in BLOB and comparing with the value $(t,v,g,id_{L1})$ stored in the Inbox Contract.
        \item \label{rule:b_burn_nonexist} Burn transaction not published to L1 blockchain but burn data is recorded in contract.
        Can be disputed by showing the transaction with index $txID$ in BLOB is not a according burn transaction. 
        \item \label{rule:b_coinbase_last} Last transaction in batch is not a fee-collecting transaction or not alone in a bracket. 
           Can be disputed by providing the type of the last transaction $txLast$ and the number $n_{tx}$ of transactions in the last bracket, and checking if $txLast$ is not a fee-collecting transaction or $n_{tx} > 1$.
        \item \label{rule:b_coinbase_other} fee-collecting transaction present in batch not at last position.
            Can be disputed by providing a transaction $tx$ and revealing its type, checking that this transaction is not in the last position, and is a fee-collecting transaction.
        \item \label{rule:b_fee_checkpoint} Intermediate fee calculation invalid. 
            Can be disputed by recalculating the collected transaction fee based on the previous transaction fee checkpoint at the beginning of the batch.
            
        \item \label{rule:b_coinbase_invalid} Invalid fee-collecting transaction. 
            Can be disputed by recalculating the commitment value with the published parameter.
            
        \item \label{rule:b_ctr_checkpoint} Intermediate coin tree root calculation invalid.
            It can be disputed by providing the front of the Merkle tree before the disputed transaction bracket and recalculating the updated root based on the outputs.

        \item \label{rule:b_ntr_checkpoint} Intermediate nullifier tree root invalid. 
            Can be disputed by providing the front of the Merkle tree before the disputed transaction bracket and recalculating the updated root based on the inputs.

    \end{enumerate}
\end{enumerate}




\subsection{Leaving Protocol}

As the final procedure for \ru, clients can leave the protocol by generating the burn transaction. After the transaction is executed, following the execution procedure defined before, the client can further retrieve the corresponding coins from the on-chain smart contract, which is shown in Algorithm~\ref{alg:bridge_l2_l1}.

A Burn Transaction consists of one coin with commitment $cm$ that the client wants to peg out with. A Burn Transaction does not have an output. The body of a burn transaction contains: (1) the L1 address $id_{L1}$ that is used to claim the burned assets on L1 through a contract. The bridge contract will check if the withdrawal request comes from this identity; (2) the private values $token, value, fee$ of the burnt coin; (3) the zk-SNARK proof to show the validity of such input; (4) a fee value $fee$ for publishing the burn transaction. Note that the fee for burn transactions is higher than for transfer transactions, because the operator must pay for L1 storage for every burn transaction. 
    This further protects against spamming L1 by financially incentivising users to first join their unused coins with the same token type by means of transfer transactions before burning them once.

\begin{algorithm}[!htp]
\caption{L2 to L1 Bridge (Contract 4)}\label{alg:bridge_l2_l1}


\Fn{$retrieve(batchId, bracketId, txId)$}{
    \Assert{$Inbox.blockFinalized(batchId)$}\; 
    $(t, v, g, id_{L1}) \gets Inbox.consumeBurnData(batchId, bracketId, txId)$\;
    \If{$id_{L1} = msg.sender$}{
    \If{$t \neq 0$}{
        $t.transfer(id, v)$; \Comment{Transfer tokens from bridge to user}
    }
    \If{$g \neq 0$}{
        $id_{L1}.transfer(g)$; \Comment{Transfer ETH from bridge to user}
    }
    }
}
\end{algorithm}





\section{Security Analysis}
\label{sec:analysis}

In this section, we prove the security and the privacy-preserving property of our \ru protocol. We separately analyze these properties in the traditional Byzantine security model and the rational security model. Due to the page limit, we only show the general proof idea here, the detailed proof can be found in Appendix~\ref{apx:Byzantine}. 


Under the Byzantine security model, we first prove that \ru realizes \emph{balance security}, \emph{liveness}. The general proof idea is that, in order to break these security properties, the adversary should be able to break the underlying assumptions, either the ideal underlying ledger or the ideal cryptographic tools that we are using in the system. We also prove \ru realizes L2 transaction privacy through the L2-IND game.

\begin{restatable}[Balance security]{theorem}{Thmbalance}
    \ru realizes balance security under the Byzantine security model.
\end{restatable}
\vspace{-1em}
\begin{restatable}[Liveness]{theorem}{Thmliveness}
    As long as there exists one honest operator and one honest verifier, the \emph{liveness} of \ru is guaranteed.
\end{restatable}
\vspace{-1.5em}
\begin{restatable}[L2 transaction privacy]{theorem}{Thmprivacy}
    \label{lem:ledger_ind}
    The \ru protocol is an L2 transaction privacy-preserving optimistic-Rollup protocol through realizing L2 transaction indistinguishability.
\end{restatable}




\section{Implementation and Evaluation}
\label{sec:implementation}

To demonstrate the feasibility of \ru, we implemented the protocol's smart contracts in Solidity v0.8.28 and the zk-SNARK circuits with Zokrates~\cite{zokrates}, and deployed the system on an Ethereum testnet. Our evaluation answers three questions: (i) what are the on-chain gas costs of each phase of the protocol, in both the optimistic and pessimistic cases; (ii) how does \ru compare asymptotically to existing zk- and optimistic-Rollup designs; and (iii) how does \ru compare empirically to a state-of-the-art privacy-preserving Rollup, Aztec, in terms of L1 cost, data-availability footprint, and client-side proving time.

Unless otherwise stated, we assume a gas base price of $8.66 \times 10^{-9}$ ETH and a BLOB gas base price of $3.52 \times 10^{-9}$ ETH, both average values taken over the previous year, and an exchange rate of $2690 : 1$ between USD and ETH.

\subsection{Experimental Setup}

We evaluate \ru along two independent paths. The on-chain costs of normal protocol operation (deployment, joining, batch publishing, and leaving) are obtained from a Python-based evaluation of our Solidity contracts, while the per-rule fraud-proof costs are measured with the Foundry gas reporter against the deployed Judge contracts. The Judge logic is split across three contracts, \texttt{TransactionJudge}, \texttt{BracketJudge}, and \texttt{BatchJudge}, with a combined bytecode size of $16{,}251$ bytes and a one-time deployment cost of $3{,}657{,}020$ gas; the total deployment of all judges, the Inbox contract, and the verifiers is $9{,}616{,}195$ gas. 

For the fraud-proof costs in Table~\ref{tab:fraudproof}, we distinguish \emph{measured} rules, whose dispute paths are exercised end-to-end against the deployed contracts, from \emph{synthetic} rules, whose execution gas is measured but whose KZG point-evaluation overhead is estimated rather than triggered on a live BLOB; the latter are marked accordingly. The reported total for each rule combines Judge-contract execution gas with the KZG overhead required to authenticate the revealed BLOB words against the batch commitment.

For the comparison in Section~\ref{subsec:aztec}, we benchmark Aztec v4.2.0-aztecnr-rc.2 on its local sandbox network. \ru's data-availability footprint and per-transfer L1 cost are amortized over a full blob ($86$ transfers, $147{,}763$ execution gas plus $131{,}072$ blob gas per batch). Aztec's figures are derived from the sandbox: the per-private-transfer data-availability size is taken from the reported \texttt{totalSizeInBytes} ($1{,}376$ bytes), and its amortized L1 cost from $2{,}401{,}898$ total L1 gas over $9$ L1 operations. The client proof time reported for \ru is the time to generate the transfer proof, whereas Aztec's figure is an end-to-end private-transfer measurement that also includes sequencing; the two are therefore not directly comparable and we treat the proving-time gap as indicative only.

\subsection{On-chain Cost}

The gas costs for each phase of \ru are shown in Table~\ref{tab:calyx_phases}. Deploying all the smart contracts introduced in Section~\ref{sec:design} costs approximately $0.43266$ ETH. This deployment is performed only once and supports long-term protocol operation; the relatively high cost arises from pre-listing all fraud-proof logic on-chain, which significantly reduces the cost of subsequent L1 interactions. Since gas prices fluctuate, deployment can additionally be scheduled when the base price is below average.

For a client to join the protocol, it must send a peg-in transaction to the L1 smart contract and create a mint transaction, which the operator later publishes to L1. The combined cost of these two steps is $0.00217$ ETH. The off-chain parts of joining and transaction preparation require no L1 publication and therefore incur zero gas cost.

We next evaluate the cost of publishing a batch. The batch size is constrained by the BLOB size; under our construction, a single batch can include up to 269 mint, 167 burn, or 86 transfer transactions. The amortized per-transfer publishing cost is $0.00002$ ETH, which also reflects the execution cost of a transaction in the optimistic case, since no fraud-proof is published when the operator behaves honestly. Finally, for a client to leave the protocol, it sends a burn transaction that the operator publishes, followed by a retrieval transaction on L1; this process costs $0.00190$ ETH in total.

\begin{table}[h]
\centering
\caption{Per-phase on-chain cost of \ru. Deployment is a one-time cost; the transfer cost is amortized over a full blob of 86 transfers.}
\label{tab:calyx_phases}
\resizebox{\columnwidth}{!}{%
\begin{tabular}{|l|r|r|r|}
\hline
Phase & Gas & ETH & USD \\ \hline
Deployment (one-time)        & 49{,}959{,}175 & 0.43266 & 1163.84 \\ \hline
Join (approve + lock + mint) & 250{,}455      & 0.00217 & 5.83    \\ \hline
Transfer (amortized per tx)  & 2{,}339        & 0.00002 & 0.05    \\ \hline
Leave (unlock + burn batch)  & 219{,}772      & 0.00190 & 5.12    \\ \hline
\end{tabular}%
}
\end{table}

\paragraph{Fraud-proof cost.} Because a verifier's willingness to challenge depends on the cost of doing so, we report the on-chain cost of \emph{every} fraud-proof rule rather than a single worst case. Table~\ref{tab:fraudproof} lists the gas cost of each rule defined in Section~\ref{sec:design}, decomposed into Judge-contract execution and the KZG point-evaluation overhead required to authenticate the revealed BLOB words. The costs range from $0.0025$ ETH-equivalent for simple structural rules (e.g., invalid transaction type or bracket count) to the worst case under rule~\ref{rule:tx_proof_zk-SNARK}, where verifying a transaction proof with the maximum of 8 inputs and outputs costs $2{,}791{,}028$ gas ($0.02417$ ETH). Crucially, every rule is disputable in a \emph{single} on-chain interaction of constant size; no rule requires multi-round localization. Since these costs are incurred only in the pessimistic case, and remain well below the operator's stake, an honest verifier is always economically able to challenge misbehavior.

\begin{table}[h]
\resizebox{\columnwidth}{!}{%
\begin{tabular}{|c|l|c|c|c|}
\hline
Rule & Description & Judge gas & + KZG & USD \\ \hline
\ref{rule:tx_type}       & Invalid TX type             & 11{,}702  & 161{,}702   & 3.77  \\ \hline
\ref{rule:tx_mint_values}& Invalid mint commitment     & 30{,}780  & 280{,}780   & 6.54  \\ \hline
\ref{rule:tx_burn_values}& Invalid burn commitment     & 11{,}665  & 511{,}665   & 11.92 \\ \hline
\ref{rule:tx_input_zk-SNARK}& Invalid input proof      & 171{,}050 & 921{,}050   & 21.46 \\ \hline
\ref{rule:tx_proof_zk-SNARK}& Invalid TX proof         & 559{,}619 & 2{,}209{,}619 & 51.48 \\ \hline
\ref{rule:tx_hash}       & Invalid TX hash             & 417{,}703 & 2{,}117{,}703 & 49.33 \\ \hline
\ref{rule:tx_double_past}& Double spend (cross-bracket)& 319{,}837 & 819{,}837   & 19.10 \\ \hline
\ref{rule:tx_double_same}& Double spend (same bracket) & 319{,}837 & 1{,}319{,}837 & 30.75 \\ \hline
\ref{rule:tx_input_ctr}  & Invalid coin tree root      & 205{,}273 & 955{,}273   & 22.25 \\ \hline
\ref{rule:txb_hash}      & Invalid bracket hash        & 39{,}559  & 239{,}559   & 5.58  \\ \hline
\ref{rule:txb_sig_count} & Signature count mismatch    & 35{,}713  & 285{,}713   & 6.66  \\ \hline
\ref{rule:txb_sig_error} & Invalid signature           & 219{,}382 & 969{,}382   & 22.58 \\ \hline
\ref{rule:txb_tx_count}  & Empty/oversize bracket      & 6{,}927   & 106{,}927   & 2.49  \\ \hline
\ref{rule:b_burn_copy}   & Burn data L1 mismatch       & 71{,}772  & 571{,}772   & 13.32 \\ \hline
\ref{rule:b_burn_nonexist}& Burn data, no matching TX  & 12{,}270  & 162{,}270   & 3.78  \\ \hline
\ref{rule:b_coinbase_last}& Last TX not fee-collecting & 20{,}259  & 170{,}259   & 3.97  \\ \hline
\ref{rule:b_coinbase_other}& Fee-collecting misplaced  & 15{,}998  & 165{,}998   & 3.87  \\ \hline
\ref{rule:b_fee_checkpoint}& Fee checkpoint invalid    & 111{,}187 & 1{,}111{,}187 & 25.89 \\ \hline
\ref{rule:b_coinbase_invalid}& Invalid fee-collecting TX& 90{,}629 & 590{,}629   & 13.76 \\ \hline
\ref{rule:b_ctr_checkpoint}& Coin tree root invalid    & 205{,}273 & 955{,}273   & 22.25 \\ \hline
\ref{rule:b_ntr_checkpoint}& Nullifier root invalid    & 220{,}000 & 720{,}000   & 16.77 \\ \hline
\end{tabular}%
}
\caption{Per-rule fraud-proof cost, decomposed into Judge-contract execution gas and total gas including KZG point-evaluation overhead for BLOB authentication. Every rule is disputable in a single constant-size interaction.}
\label{tab:fraudproof}
\end{table}

\subsection{Asymptotic Comparison}

We compare \ru asymptotically with the SNARK--aggregation--based zk-Rollup Aztec and the Arbitrum optimistic-Rollup. Assuming a batch of $n$ transactions, the comparison is shown in Table~\ref{table:comparison}.

Since all three protocols realize \emph{data availability}, every executed transaction is published in the BLOB space of L1, giving $O(n)$ on-chain cost. For non-BLOB L1 cost, Aztec requires the on-chain contract to verify the aggregated proof produced by the operator~\cite{aztec}; although no formal analysis of their aggregation scheme is available, prior work on SNARK aggregation~\cite{liu2023snarkfold} shows that, absent preprocessing, verification requires $O(\log n)$ computation. Arbitrum's bisection scheme~\cite{kalodner2018arbitrum} introduces $O(\log n)$ on-chain challenge interactions in the pessimistic case. \ru, under optimistic execution, requires only $O(1)$ on-chain checkpointing, like Arbitrum; but owing to its one-step fraud-proof scheme, it further reduces pessimistic verification to $O(1)$ interaction time and proof size. Regarding privacy, Arbitrum offers none, as all transactions appear on-chain in plaintext, whereas both Aztec and \ru achieve transaction-level privacy via zk-SNARKs.

\begin{table}[h]
\resizebox{\columnwidth}{!}{%
\begin{tabular}{|c|c|cc|c|}
\hline
      & BLOB L1 cost & \multicolumn{2}{c|}{non-BLOB L1 cost}                & Privacy-preserving        \\ \hline
Aztec & $O(n)$         & \multicolumn{2}{c|}{$O(\log n)$}        & $\checkmark$ \\ \hline
\multirow{2}{*}{Arbitrum}           & \multirow{2}{*}{$O(n)$} & \multicolumn{1}{c|}{Optimistic} & Pessimistic & \multirow{2}{*}{$\times$}     \\ \cline{3-4}
      &              & \multicolumn{1}{c|}{$O(1)$} & $O(\log n)$ &                           \\ \hline
\multirow{2}{*}{\ru} & \multirow{2}{*}{$O(n)$} & \multicolumn{1}{c|}{Optimistic} & Pessimistic & \multirow{2}{*}{$\checkmark$} \\ \cline{3-4}
      &              & \multicolumn{1}{c|}{$O(1)$} & $O(1)$                    &                           \\ \hline
\end{tabular}%
}
\caption{Asymptotic comparison of L1 cost and privacy, for a batch of $n$ transactions published on L1.}
\label{table:comparison}
\end{table}

\subsection{Empirical Comparison with Aztec}
\label{subsec:aztec}

To ground the asymptotic analysis, we empirically compare \ru against Aztec (v4.2.0) on a local sandbox network. Table~\ref{tab:aztec_comparison} summarizes the per-transfer L1 cost, data-availability footprint, client proving time, and finality of the two systems.

In the optimistic case, \ru publishes a private transfer for an amortized cost of about $2{,}339$ gas (\$0.05) over a full blob of 86 transfers. The data-availability footprint is comparable between the two systems, about $1{,}475$ bytes per two-input/two-output transfer in \ru versus $1{,}376$ bytes in Aztec, yielding a similar blob capacity of 86 versus 92 transfers per blob, and a per-transfer full-blob L1 cost of roughly \$0.055 for \ru against \$0.081 for Aztec. Unlike Aztec, \ru incurs no on-chain verification in the optimistic case, paying an on-chain fraud-proof cost (worst case \$65.02) only when an operator misbehaves and is challenged.


\begin{table}[h]
\centering
\caption{Empirical comparison of \ru and Aztec (v4.2.0-aztecnr-rc.2, sandbox local network).}
\label{tab:aztec_comparison}
\resizebox{\columnwidth}{!}{%
\begin{tabular}{|l|r|r|}
\hline
Metric & \ru & Aztec \\ \hline
Publihsed data per 2-in/2-out transfer       & 1{,}475 B   & 1{,}376 B \\ \hline
Transfers per blob               & 86          & 92        \\ \hline
L1 batch gas (per batch)    & 147{,}763   & $\sim$265{,}000 \\ \hline
L1 cost per transfer (full blob) & \$0.0545    & \$0.0806  \\ \hline
Fraud proof (worst case)         & \$65.02     & N/A (ZK)  \\ \hline
Finality                         & $\sim$7 days & $\sim$minutes \\ \hline
Deployment (total)               & 9{,}616{,}195 gas & (pre-deployed) \\ \hline
\end{tabular}%
}
\end{table}

\subsection{Discussion}

Taken together, these results show that \ru delivers multi-token privacy at an optimistic per-transfer cost competitive with, and slightly below, a state-of-the-art privacy-preserving zk-Rollup, while keeping pessimistic on-chain verification constant-size. The principal cost of the optimistic design is delayed finality. Since the highest gas costs arise only when an operator misbehaves and a verifier challenges, and since \ru achieves L2 transaction privacy simultaneously, we conclude that its overall overhead is acceptable for privacy-preserving multi-token payments.

\section{Extension}

\textbf{\indent Apply to CBDC.} \ru is initially designed for privacy-preserving, scalable multi-token exchange on the Ethereum platform, but we believe the underlying idea can be applied to broader scenarios. For instance, Central Bank Digital Currencies (CBDCs)~\cite{ozili2023cbdc} have gained increasing attention in recent years as a convergence of blockchain systems and centralized payment infrastructures. Exploring how to achieve both privacy and efficiency in such a setting using the principles of \ru represents an interesting research direction.

\textbf{Transfer among different token types.} In the current design of \ru, each transaction requires that all inputs and outputs belong to a single token type, ensuring that no loss occurs due to dramatic exchange rate fluctuations in the Ethereum token market~\cite{okx2025neiroethcrash}. We believe \ru can be extended to support transfers across different token types if a fixed globally agreed-upon exchange rate is encoded at the time the smart contract is deployed.


\textbf{Extend to contract execution.} The \ru protocol is initially designed for transaction execution, and thus the fraud-proof rules and zk-SNARK circuits are tailored to provide security guarantees in this setting. If detailed security guarantees, fraud-proof schemes, and corresponding zk-SNARK circuits can be defined for specific smart contract executions, we believe the design of \ru can be further extended to support contract execution scenarios.

\section{Conclusion}

In this paper, we propose the very first privacy-preserving multi-token optimistic-Rollup protocol, termed \ru. \ru leverages zk-SNARKs to protect L2 transaction privacy, combined with a carefully designed interaction model with the L1 blockchain to achieve state consistency across layers and to support efficient L1 fraud-proof interaction. Beyond presenting the protocol, we analyze its security under both the Byzantine and rational adversary models. We also formally define L2 transaction privacy for Rollup protocols and provide a corresponding analysis. Finally, we implement the on-chain components of \ru, demonstrating its practicality.

\section*{Acknowledgement}
This work was partially supported by the Vienna Science and Technology Fund (WWTF) through the project 10.47379/ICT22045; by the Austrian Science Fund (FWF) through the
SpyCode SFB project F8510-N and F8512-N; by the European Research Council (ERC) under the European Union’s Horizon 2020 research (grant agreement 101141432-BlockSec).

\bibliographystyle{plain}
\bibliography{reference}

\appendix

\section{Proof for the Byzantine security model }
\label{apx:Byzantine}

\subsection{Balance security}

The main idea for proving balance security in \ru is to show that the adversary cannot affect the correctness of each participant's L2 and L1 states. In the following, we demonstrate that across the three phases of \ru—protocol joining, L2 transaction execution, and protocol leaving—the adversary is unable to influence honest participants.

\paragraph{Protocol joining.}The security goal for this first step can be concluded as following Lemma:

\begin{lemma}[Correct joining]
For any deposit on the L1 blockchain, a corresponding \ru L2 state is initialized through a mint transaction.
\end{lemma}

\begin{proof}
Assume the lemma does not hold. Then the violation must occur in one of the following three cases:

\begin{enumerate}
\item A mint transaction occurs on $L_2$ without a corresponding $LOCK$ event on $L_1$. In such a case, 
This misbehavior can be challenged by $V$ using rule~\ref{rule:tx_mint_values}, which checks the revealed output values of the mint transaction against the claimed bridge nonce. Since the bridge also records the block number at which the nonce was created, the Judge can verify the "happened earlier" relation.

\item A $LOCK$ event is claimed twice with two different mint transactions. In such a case, as mint nonces are treated as nullifiers, rules~\ref{rule:tx_double_past} and~\ref{rule:tx_double_same} prevent a mint nonce from being used more than once in mint transactions.

\item A mint transaction occurs on $L_2$ with a corresponding $LOCK$ event on $L_1$, but is challenged and invalidated by a malicious $V$.
The malicious $V$ can only challenge through rules~\ref{rule:tx_mint_values}, \ref{rule:tx_double_past}, and \ref{rule:tx_double_same}.
If $V$ is able to successfully challenge through rule~\ref{rule:tx_mint_values}, then the $LOCK$ event triggered by a committed deposit transaction would be reverted, which contradicts the assumption of an ideal underlying L1 blockchain.
If $V$ is able to successfully challenge through rules~\ref{rule:tx_double_past} or~\ref{rule:tx_double_same}, then the adversary must be able to create a hash collision without knowing the secret input of the honest joining client, which contradicts the assumption of an ideal hash function.
\end{enumerate}

Conclusively, in \ru the correctness of protocol joining cannot be broken by the adversary.
\end{proof}

\paragraph{L2 transfer transaction execution.}According to \ru as introduced in~\ref{sec:design}, a regular L2 transfer transaction is considered executed only after it is published on the L1 blockchain and no fraud-proof is submitted by a verifier within a specified time window. We state the security property of L2 transfer transaction execution in \ru with the following lemma:

\begin{lemma}[Correct L2 transfer transaction execution]
    No honest participants will lose coins during the execution of L2 transfer transactions in \ru. 
\end{lemma}

\begin{proof}
    We prove the lemma by enumerating the possible violations and showing that each contradicts our underlying assumptions.
    
    Since \ru adopts the UTXO model for L2 transactions, the only way an honest client could lose a coin is if the coin were spent without the client’s permission. For this to happen, either the adversary would need to forge a valid signature or learn the hidden secret from the zk-SNARK proof. Both cases contradict our assumptions of an ideal signature scheme and an ideal zk-SNARK scheme.
\end{proof}

\paragraph{Protocol leaving.}As the final procedure, our protocol must guarantee that each honest client can leave the protocol with its corresponding L2 state. The adversary may attempt to influence this procedure in two ways. First, the adversary could try to disrupt the execution of a client’s leave request. Second, the adversary could attempt to execute invalid transactions, thereby affecting the token storage in the \ru contract and indirectly impacting the success of an honest client’s departure. In the following, we analyze these two situations separately.

For the first possible influence, we have the following Lemma:

\begin{lemma}[Correct execution of leaving request]
    In \ru, the leaving request from the honest client will not be invalidated once published on the L1 blockchain.
\end{lemma}

\begin{proof}
    First, the adversary could attempt to disrupt the procedure by publishing an invalid burn transaction. However, a burn transaction must include a revealed input that is checked by rule~\ref{rule:tx_burn_values}. In addition, the L1 beneficiary address is verified by rule~\ref{rule:b_burn_copy}. Therefore, as long as an honest verifier exists, any such violation will be detected and challenged.
    
    Second, the adversary could attempt to influence the procedure by corrupting the verifier and trying to invalidate a valid burn transaction with a fraudulent challenge. In this case, the verifier can only invoke rules~\ref{rule:b_burn_copy} and~\ref{rule:b_burn_nonexist}. Once the burn transaction has been published on the L1 blockchain, the checks in these rules will succeed unless the consistency of the underlying L1 blockchain is broken, which contradicts our assumption.

    Finally, without a correct burn transaction included in the BLOB and the corresponding information, such as the beneficiary address, committed to the smart contract, no party is able to retrieve coins from the contract.
\end{proof}

For the second possible influence, we have the following Lemma:

\begin{lemma}[Correct token value change]
In \ru, the adversary cannot alter the total deposited token amount through invalid transactions.
\end{lemma}

\begin{proof}
    The adversary could attempt to change the total deposited token amount in the following ways: (1) double-spending tokens; (2) creating transactions where the output value does not match the input value; (3) generating multiple mint or burn transactions.

For the first case, rule~\ref{rule:tx_input_zk-SNARK} ensures that the input zk-SNARK proof is correct. A correct input zk-SNARK proof guarantees that the deterministic serial number (nullifier) is revealed for each input. All nullifiers ever revealed will appear in the nullifier tree committed to by all \textit{following} transaction brackets. Thus, rule~\ref{rule:tx_double_past} prevents the reuse of nullifiers across different transaction brackets, while rule~\ref{rule:tx_double_same} prevents double spending within the same bracket.

For the second case, rule~\ref{rule:tx_proof_zk-SNARK} ensures that every zk-SNARK proof in transfer transactions is correct, and a correct zk-SNARK proof guarantees that each transfer transaction balances. Furthermore, every input to a transfer transaction must be correctly derived from a coin output, which is also validated by the input zk-SNARK proof, whose correctness is guaranteed by rule~\ref{rule:tx_input_zk-SNARK}. Therefore, as long as there exists one honest verifier, any imbalance can be detected and disputed.

For the third case, if a $LOCK$ event is claimed twice, mint nonces are treated as nullifiers, and rules~\ref{rule:tx_double_past} and~\ref{rule:tx_double_same} prevent the reuse of mint nonces in multiple mint transactions. If a nullifier is revealed in a transaction bracket $txb$, then it will appear in all nullifier trees of later brackets $txb'$ with $txb < txb'$. Similarly, if a burn transaction were to be claimed twice, this would fail since the bridge contract on L1 clears the reference to that specific burn transaction once claimed. Therefore, as long as at least one honest verifier exists, such violations will be detected and successfully disputed.
\end{proof}

Based on the conclusions above, we can have the following result:

\Thmbalance*

\subsection{Liveness}

Since the execution of a transaction in \ru requires two steps: the operator publishing transaction batches on the L1 blockchain and the subsequent verification by a verifier, we prove that liveness is guaranteed in \ru under the following assumptions:

\Thmliveness*

\begin{proof}

First, as long as there exists at least one honest operator in \ru, the transactions of honest clients will eventually be published on the L1 blockchain once they are delivered to the operator under our asynchronous communication model. A round-robin operator election can further ensure that the honest operator has the opportunity to publish batches on-chain. Any batch published by a dishonest operator will be disputed by the honest verifier and therefore will not affect the execution of transactions from honest clients.
   
   Then, the only way for the adversary to influence liveness by corrupting the verifier is to provide a valid fraud-proof $\pi$ against correct behavior from honest clients and honest operators. After a deposit transaction, transaction batch, or checkpoint has been committed on the L1 blockchain, if the verifier were able to generate a proof of non-existence of such actions, it would require reverting the underlying L1 blockchain, which contradicts the assumption of an ideal blockchain. Another possible way to construct a fraud-proof would be to either create a collision in the hash function, compute a preimage of a hash value or an encrypted message, or generate a valid zk-SNARK proof for an invalid input. All of these contradict our assumption of ideal cryptographic tools. Thus, a verifier cannot influence liveness.
\end{proof}

\subsection{Privacy}
\label{apx:privacy}

To prove the privacy of \ru, we need to prove that there exists no PPT adversary $\mathcal{A}$ that can win the L2-IND game with non-negligible probability. 

\subsubsection{Game and oracle definition}

First of all, the oracles used in the game are defined as follows:

\textbf{Ledger Oracle:} A ledger oracle $\mathcal{L}_1$ allows for two types of requests:
\begin{itemize}
    \item $WriteOnL1(tx)$ takes a valid transaction $tx$ and append to the transaction list it maintained.

    \item $ReadOnL1()$ returns all the committed L1 blockchain transaction list and states.
\end{itemize}

Furthermore, since the transaction list of \ru is also published on the L1 blockchain, here we similarly define a L2 ledger oracle $\mathcal{L}_2$ that is connected to $\mathcal{L}_1$ and also takes two types of requests:
\begin{itemize}
    \item $CommitBatch(batch)$ takes an input batch and appends to the L1 transaction list first, and then to the L2 transaction list if no fraud-proof response is received by interacting with the verifier oracle $\mathcal{V}$.

    \item $ReadOnL2()$ returns all the committed L2 transaction list and states.
\end{itemize}

\textbf{Client Oracle:} A client oracle $\mathcal{H}$ needs access to both the L1 and the L2 ledger oracle and models the behaviour of clients and allows for the following queries:
\begin{itemize}
    \item $CreateAddress(pp)$ creates a new L1 address $id_{L1}$ (generated from $pk_{sig}$) and corresponding key pair $pk=(pk_{sig},pk_{enc},pk_{coin})$ and $sk=(sk_{sig},sk_{enc},sk_{coin})$ based on public parameter $pp$. The mapping of the L1 address and the key pair is stored in the oracle. Noted that the key pair can also be generated separately.

    
    \item $PegIn(id_{L1}, a)$ takes the L1 address $id_{L1}$, asset $a$, the encoded contract address $B$ and prepare a pegin transaction $tx_{pegin}$ and send to $\mathcal{L}_1$ through request $WriteOnL1(tx_{pegin})$. Then, it creates a fresh nonce $n$ and records the mapping. And finally returns $(n, tx_{pegin})$.
    
    \item $Mint(id_{L1}, a, n)$ checks on the L1 oracle through $ReadOnL1()$ request whether $n$ is a bridge nonce that was awarded for locking $a$ assets and checks that no mint transaction for that bridge nonce was created by this oracle.
    If a check fails, it returns $\bot$.
    Then, it calls $CreateAddress$ to create a new coin keypair $(pk_{coin}, sk_{coin})$. Based on $(pp, n, a, pk_{coin})$, use the hash and commitment scheme, generate and return a new mint transaction $tx_{mint}$ containing coin $c$.
    Coins created by this oracle are added to the coins array $COINS$, and their index in this array is denoted by $cid(c)$.
    Finally, creates a transaction bracket $txb$ for mint transaction and returns $(cid(c), txb)$.
    
    \item $RevealValues(cid,sk)$ takes a coin index $cid$, and checks whether $sk$ proves the ownership of the coin with the decryption and commitment scheme. If check passes, the oracle returns the private values $t, v, g$ of $COINS[cid]$. Otherwise, it returns the random value attached to the coin's transaction.
    
    \item $Transfer$ takes a list of coin commitments $(cid_0, cid_1, ..., cid_k)$ and a list of intended outputs $((pk_0, t_0, v_0, g_0), (pk_1, t_1, v_1, g_1), ..., (pk_n, t_n, v_n, g_n))$.
    If any of the referenced coins do not exist through $ReadOnL2()$ check, a transfer or burn transaction has already been created for a coin, or the query cannot be fulfilled with a correct transfer transaction, the oracle returns $\bot$.
    It attempts to create a transfer transaction $tx$ that spends all input coins and creates the requested outputs.
    All output coins created are added to the $COINS$ array and their indices are returned alongside $txb$, which is a transaction bracket that contains $tx$.

    \item $Burn(cid,id_{L1},sk)$ takes a coin index $cid$ and a L1 address $id_{L1}$.
    If a coin $c$ exists at this position for which no transfer or burn transaction has yet been created, a bracket containing a burn transaction that burns this coin and specifies $id_{L1}$ as beneficiary address is created and returned.
    
    \item $PegOut(cid,id_{L1})$ takes a coin index $cid$ and a L1 address.
    It creates and returns a L1 transaction that claims these assets if $cid$ is a coin for which a burn transaction is already placed on L2, which is checked through request $ReadOnL2()$.
    Only one such transaction can be created for the same burn transaction.
\end{itemize}

Note that this oracle can also be queried with a list of mint, transfer or burn transactions, in which case a single transaction bracket is returned that contains all created transactions.
If any transaction fails to be created, then no bracket is returned at all.
Further note that $\mathcal{H}$ will automatically run the coin reception algorithm $Receive$ with all its maintained identities. 

\textbf{Operator Oracle:} An operator oracle $\mathcal{S}$ needs access to both the L1 and the L2 ledger oracle and models the behaviour of a live, honest operator that can be queried to create batches from transaction brackets.

It allows for the following queries:
\begin{itemize}
    \item $AddBracket(txb)$ takes a bracket from clients and attempts to add it to its mempool by updating $M' = M + txb$, and checks that the resulting blob serialization does not exceed the maximum BLOB size.
    If it doesn't exceed the size, then it sets $M := M'$, i.e. it updates its mempool.
    Otherwise, it creates a new L1 identity $(pk, sk)$ using $CreateAddress(pp)$, takes the brackets from its mempool $M$, and creates a new valid batch $b$. 
    This batch is then sent to $\mathcal{L}_2$ through $CommitBatch$ request, and the mempool is reset: $M := ()$.
\end{itemize}

\textbf{Verifier Oracle:} The verifier oracle $\mathcal{V}$ takes the following request:

\begin{itemize}
    \item $CommitBatch(batch)$ takes an appended $batch$ and will return a corresponding fraud-proof if there exists any invalid transaction there, otherwise outputs nothing.
\end{itemize}

Then we can formally define the L2-IND game as follows:

\begin{definition}(L2-IND game)
    \label{def:lind}
    The L2-IND game is a round-based game played between a challenger $\mathcal{C}$ and an adversary $\mathcal{A}$ that proceeds as follows: First is the setup phase, 
    $\mathcal{C}$ publicly creates $pp \gets Setup(1^\eta)$ with security parameter $\eta$, and then it creates the following:
    \begin{itemize}
        \item Two L1  ledger oracles: $L_{1;0}$ and $L_{1;1}$.
        \item Two L2 ledger oracles: $L_{2;0}$ and $L_{2;1}$.
        \item Two client oracles: $\mathcal{H}_0$ and $\mathcal{H}_1$.
        \item Two operator oracles: $\mathcal{S}_0$ and $\mathcal{S}_1$.
        \item Two verifier oracles: $\mathcal{V}_0$ and $\mathcal{V}_1$.
    \end{itemize}
    
    Furthermore, $\mathcal{C}$ privately samples a random bit $b \gets_\$ \{0,1\}$. Then in each round, $\mathcal{A}$ sends to $\mathcal{C}$ a pair of the same type of queries $(Q_0, Q_1)$. $\mathcal{C}$ first checks if $Q_0$ and $Q_1$ are publicly consistent, which is defined in Definition \ref{def:lind_pc}. If they are not, then $\mathcal{C}$ rejects this query and asks for the next one. Else, depending on the type, the following happens:
    \begin{itemize}
        \item If the query type is $CommitBatch$, then
        $\mathcal{C}$ forwards $Q_0$ to $L_{2;0}$ and $Q_1$ to $L_{2;1}$.
        Then, $\mathcal{C}$ asks $\mathcal{V}_0$ and $\mathcal{V}_1$ if the appended batches are correct.
        If at least one of them outputs a fraud-proof, $\mathcal{C}$ outputs $0$ (indicating that the Adversary lost the game)
        \item If the query type is $WriteOnL1$, $PegIn$ and $PegOut$, then $\mathcal{C}$ forwards $Q_0$ to $L_{1;0}$ and $Q_1$ to $L_{1;1}$.
        If either ledger rejects the transaction, $\mathcal{C}$ outputs $0$.
        \item The Adversary may end the game by sending a guess $b_\mathcal{A}$.
        $\mathcal{C}$ outputs $1$ iff $b_\mathcal{A} = b$, else $0$.
        \item Otherwise, the query is a query to the client oracle. $\mathcal{C}$ forwards $Q_0$ to $\mathcal{H}_{b}$, which results in some information output like a newly created public key and newly created coin indices, an L1 transaction $tx_{L1}$, or an L2 transaction $txb$.
        $\mathcal{C}$ forwards any L1 transaction to $L_{1;b}$.
        $\mathcal{C}$ forwards any L2 transaction bracket $txb$ to $\mathcal{S}_b$.
        Any information output is returned to $\mathcal{A}$.

        $\mathcal{C}$ performs the same procedure with $Q_1$, but instead of forwarding it to $\mathcal{H}_{b}$, it forwards it to the other honest client oracle $\mathcal{H}_{1-b}$.
        Any L1 transaction is forwarded to $L_{1;1-b}$, any L2 transaction bracket to $\mathcal{S}_{1-b}$.
    \end{itemize}
    The Adversary may also read the state of any ledger oracle at any time. Finally $\mathcal{A}$ sends the $guess$ request with $b_{\mathcal{A}}$ to $\mathcal{C}$ and is considered to win the game if $b_{\mathcal{A}} = b$ with non-negligible probability greater than $\frac{1}{2}$.
\end{definition}

\begin{definition}(Public Consistency)
    \label{def:lind_pc}
    Given two queries $Q_0, Q_1$ of the same type 
    we call them \textbf{publicly consistent} according to the following rules:
    \begin{enumerate}
        \item \label{enu:lind_pc_ca} All $CreateAddress$ queries are publicly consistent, but we require that both oracles create the same address.
        \item \label{enu:lind_pc_lmbu} If the type is $WriteOnL1, CommitBatch,PegIn_{L1}, Mint,\allowbreak Burn$ or $PegOut_{L1}$, these queries are publicly consistent if and only if all public values are equal in both queries.
        \item If the type is $Transfer$, the following rules determine if the queries are publicly consistent:
        \begin{enumerate}
            \item \label{enu:lind_pc_io} The size and total transaction value of both queries must be equal, i.e., the number of inputs and outputs of $Q_0$ must match the number of inputs and outputs of $Q_1$.

            \item \label{enu:lind_pc_acc} Both queries $Q_0, Q_1$ have the same results in the validity checking conducted by the oracle. 
            \item \label{enu:lind_pc_pub} The public information (e.g., the fee) on both queries must be equal.
            \item Both queries must also be equivalent with respect to the information available to $\mathcal{A}$, which means
            \begin{enumerate}
                \item \label{enu:lind_pc_known} For every output of $Q_0$ that specifies a recipient address which is controlled by the adversary (i.e., it was not created by $CreateAddress$), then the assets of this output must match in both queries. 
                The same applies to $Q_1$, respectively.
                \item \label{enu:lind_pc_insert} If an input $i_j$ of $Q_0$ references (in the ledger $L_{2;b}$) an output contained in a transaction that was created by the Adversary (i.e., it was appended to the ledger by means of a $CommitBatch$ query), then the corresponding input $i'_j$ in $Q_1$ must reference (in the ledger $L_{2;1-b}$) a coin commitment that also appears in a transaction posted via a $CommitBatch$ query.
                Furthermore, the assets of these inputs $i_j$ and $i'_j$ must be equal in both queries.
                The same applies to $Q_1$, respectively.
            \end{enumerate}
        \end{enumerate}
        Recall that the adversary is allowed to send a list of the above queries that request L2 transactions, in which case all resulting transactions are packed inside a single bracket.
        In this case, these lists of queries are publicly consistent iff the number of queries is the same in both lists, and every corresponding pair of queries is publicly consistent.
    \end{enumerate}
\end{definition}

Finally, we can prove the \emph{L2 transaction privacy} of \ru by conducting the L2-IND experiment.

\Thmprivacy*

\begin{proof}
    We assume that there is an adversary $A$ running in probabilistic polynomial time that has a non-negligible advantage $p_A$ in the L2-IND game.
    Using a series of modifications for which we prove are at most distinguishable with negligible probability from the real L2-IND game we arrive at a simulation $\mathcal{G}_S$.
    No oracle answer from the simulation $\mathcal{G}_S$ depends on the bit $b$, thus the advantage any adversary $A'$ can have in $\mathcal{G}_S$ is precisely 0.
    This results in a contradiction, thus completing the proof that no adversary $A$ can exist.

    The behaviour of the verifier $\mathcal{C}$ in the simulation $\mathcal{G}_S$ differs from its behaviour in the L2-IND game as follows:
    When running the $Setup$ algorithm to create the public parameters $pp$, $\mathcal{C}$ stores the trapdoor information $trap_{input}, trap_{tx}$ for the input and transaction zk-SNARK proof schemes.
    $\mathcal{C}$ will use that information to be able to provide zk-SNARK proofs for any witness/statement pair $(x, w)$ that will be accepted by $\mathcal{V}_0$ and $\mathcal{V}_1$, regardless whether $(x, w)$ are in the relation or not.
    When $\mathcal{C}$ receives a publicly consistent query pair $Q_0, Q_1$ that it would normally forward to $\mathcal{H}_b, \mathcal{H}_{1-b}$, it instead forwards it to $\mathcal{M}_{0}, \mathcal{M}_{1}$. 
    Such an oracle $\mathcal{M}$ behaves as follows:
    \begin{itemize}
        \item If the query type is $CreateAddress$, it runs $(pk, sk) \gets CreateAddress(pp)$ to obtain an address.
        Recall that $pk = (pk_{sig},pk_{enc},pk_{coin})$, where $pk_{coin}$ is the public key used in the coin commitment and $pk_{enc}$ is the public encryption key used to send private coin data to the receiver. 
        Then, it randomly samples $pk_{coin}'$ from the same distribution from which the coin public keys are sampled, and sets $pk' = (pk_{coin}', pk_{enc})$.
        It stores $sk$ and returns $pk'$.
        \item If the query type is $Transfer$, the way it calculates the inputs for the resulting transfer transaction differs from the real algorithm $Transfer$ in the following way:
        \begin{itemize}
            \item Randomly sample a serial number $sn'$ that is equal in length to real serial numbers and use $sn'$ instead of the real $sn$.
            \item Randomly sample a public key authorizer $pk_{auth}'$ that is equal in length to real public key authorizers and use $pk_{auth}'$ instead of the real $pk_{auth}$. 
            \item The remaining values are constructed as defined in the $Transfer$ algorithm.
            \item The resulting input $I'$ generated from $sn'$, $ pk'_{auth}$ is not a valid input, thus an honest zk-SNARK prover cannot create a valid input proof for $I'$.
            However, since $\mathcal{M}$ is given access to $trap_{input}$, it is able to forge a valid zk-SNARK proof $\pi_{I}'$ that will be accepted by the Verifier.
        \end{itemize}
        Also, the way $\mathcal{M}$ calculates the outputs for the resulting transfer transaction differs from the real $Transfer$ algorithm: 
        \begin{itemize}
            \item If the query requests that the coin is sent to an address $pk =: (pk_{sig}, pk_{enc},pk_{coin})$ maintained by $\mathcal{M}$, a random coin commitment value $cm'$ is generated.
            Furthermore, the encrypted extra data $d_{enc}$ is obtained by encrypting random data of the correct length under $pk_{enc}$.
            \item If the query requests that the coin is sent to an address not maintained by $\mathcal{H}$ (i.e., it is under control of the Adversary), $\mathcal{M}$ correctly constructs the output as specified in the $Transfer$ algorithm.
            \item The remaining values are constructed as defined in the $Transfer$ algorithm.
        \end{itemize}
        Clearly, the resulting transaction $tx'$ generated from $\mathcal{H}$ is not a valid transaction.
        However, since $\mathcal{M}$ is given access to $trap_{tx}$, it is able to forge a valid zk-SNARK proof $\pi$ that will be accepted by the Verifier.
    \end{itemize}
    Specifically note that in this simulation, no query response given to the $\mathcal{A}$ depends on bit $b$.
    Thus, any advantage the Adversary can have in this simulation is $0$. We proceed in showing that the Adversary can distinguish such a simulation $S$ from a real execution of the game at most with negligible advantage by providing a series of modifications. Denote with $\mathcal{G}_{real}$ the real game, and $\mathcal{G}_{S}$ the simulation $S$.
    \paragraph{From $\mathcal{G}_{real}$ to $\mathcal{G}_1$.}This modification consists of faking all zk-SNARK proofs using the trapdoor information obtained from the $Setup$ algorithm.
    This modification does not result in a different distribution of the zk-SNARK proofs, since we assumed that the zk-SNARK scheme used is zero-knowledge.

    \paragraph{From $\mathcal{G}_1$ to $\mathcal{G}_2$.}This modification consists of replacing the encrypted data in outputs that are maintained by an honest client with a pure random value.
    The adversary does not have access to the honest client's secret key, and since by our assumption on the encryption scheme used, the advantage that the adversary has in distinguishing ciphertexts from random is at most negligible.

    \paragraph{From $\mathcal{G}_2$ to $\mathcal{G}_3$.}This modification consists of replacing hash digests (coin public keys, serial numbers, public key authorizers) with random values, as defined in the above definition of $\mathcal{G}_S$.
    Since we assume that our hash function is a PRF, the advantage that the adversary has in distinguishing hash digests from random values is at most negligible.

    \paragraph{From $\mathcal{G}_3$ to $\mathcal{G}_S$.}This modification consists of replacing some coin commitments with random values, as defined in the above definition of $\mathcal{G}_S$. 
    Note that coin outputs already no longer depend on the bit $b$, since any two outputs created in a pair of mint queries must commit to the same public values in order for the queries to be publicly consistent, to the random coin public key (introduced by $\mathcal{G}_3$) and random coin identifier.
    Since our construction uses the hash function as a hiding commitment scheme (by including random values not known to the adversary in the pre-image), the advantage that the adversary has in distinguishing commitments from random values is at most negligible.
\end{proof}

\end{document}